\documentclass[conference]{IEEEtran}
\IEEEoverridecommandlockouts

\usepackage{cite}
\usepackage{amsmath,amssymb,amsfonts}
\usepackage{algorithmic}
\usepackage{graphicx}
\usepackage{textcomp}
\usepackage{xcolor}
\usepackage{algorithm}
\usepackage{array}
\usepackage[caption=false,font=normalsize,labelfont=sf,textfont=sf]{subfig}
\usepackage{textcomp}
\usepackage{stfloats}
\usepackage{url}
\usepackage{verbatim}
\usepackage{cite}
\usepackage{amssymb}
\usepackage{bm}
\usepackage{cuted}
\usepackage{placeins}

\usepackage[utf8]{inputenc}
\usepackage{lineno}
\usepackage{cases}
\usepackage{float}
\usepackage{bbold}
\usepackage{siunitx}

\usepackage{amsthm}
\newtheorem{proposition}{Proposition}

\def\BibTeX{{\rm B\kern-.05em{\sc i\kern-.025em b}\kern-.08em
    T\kern-.1667em\lower.7ex\hbox{E}\kern-.125emX}}
\begin{document}

\title{Joint Precoding and Phase-Shift Optimization for Beyond-Diagonal RIS-Aided ISAC System\\
\thanks{\scriptsize This work was supported by the National Natural Science Foundation of China (Grant No. 62071275). 
	Corresponding author: Ju Liu (juliu@sdu.edu.cn).}
	\thanks{$\dag$ Equal contribution.}
}

\author{\IEEEauthorblockN{
		Xuejun Cheng$^{\S\dag}$,
		Qian Zhang$^{\S\dag}$, 
		Yuhui Jiao$^\S$, 
		Shiyao Guo$^\S$,
		Xiaotong Xu$^\S$,
		Guanghui Luo$^\ddagger$,
		Ju Liu$^{\S\ast}$}
\IEEEauthorblockA{$^\S$School of Information Science and Engineering, Shandong University, Qingdao, China \\
$^\ddagger$National Mobile Communications Research Laboratory, Southeast University, Nanjing, China\\
Emails: \{chengxuejun, qianzhang2021, yuhuijiao2024, shiyaoguo, 202332705\}@mail.sdu.edu.cn,\\ guanghuiluo@seu.edu.cn, juliu@sdu.edu.cn}}

\maketitle

\begin{abstract}
Beyond-diagonal reconfigurable intelligent surfaces (BD-RIS) can realize interconnections among reflecting elements through an impedance network, thereby providing a new approach to improving the performance of integrated sensing and communication (ISAC) systems.
This paper investigates the optimization problem in BD-RIS-aided multiuser ISAC system, aiming to enable a flexible trade-off design between communication and sensing performance.
Specifically, we propose an optimization framework that jointly integrates multiuser interference management and sensing beam gain approximation. By jointly optimizing the precoding vector and the RIS phase-shift matrix, the proposed framework improves the multiuser communication sum rate through the interference management method and enhances the sensing performance through the beam gain approximation method.
For the resulting non-convex weighted optimization problem, we employ an alternating optimization (AO) algorithm to decouple it into two subproblems, namely precoding vector optimization and phase-shift matrix optimization, where each step admits a closed-form solution.
Simulation results demonstrate that the proposed BD-RIS-aided ISAC system can achieve significant improvement in the trade-offs between communication and sensing performance than the traditional diagonal RIS, verifying the effectiveness of the proposed optimization framework.
\end{abstract}

\begin{IEEEkeywords}
Beyond-diagonal reconfigurable intelligent surfaces (BD-RIS), integrated sensing and communication (ISAC), phase-shift optimization, closed-form solutions.
\end{IEEEkeywords}

\section{Introduction}

The development of sixth-generation (6G) mobile communication networks has given rise to emerging applications such as autonomous vehicles and smart cities, which pose higher requirements for real-time environmental sensing~\cite{liu2022survey}. 
To meet this demand, wireless networks need to achieve effective sensing of the surrounding environment while maintaining high communication efficiency. Integrated sensing and communication (ISAC) technology, by sharing spectrum and hardware resources, improves spectrum and energy efficiency while effectively reducing system deployment costs, and is thus regarded as one of the most promising key technologies in 6G wireless networks~\cite{9737357,10772590,11063223}.

The wireless propagation environment has a significant impact on the performance of ISAC systems. Severe occlusion in scenarios such as dense urban areas often leads to the absence of line-of-sight paths, which causes degradation in communication quality and limited sensing accuracy~\cite{10579074}. Reconfigurable intelligent surface (RIS), as an emerging low-cost technology, provides a feasible solution to mitigate this problem. By dynamically adjusting the phase-shifts of reflecting elements, RIS can reconfigure wireless channels, thereby enhancing communication or sensing capabilities~\cite{10989512,11139112,10623806,10057422}. Based on this advantage, RIS-aided ISAC systems have received extensive attention from the academic community in recent years. 
In \cite{10858764}, Wang {\it et al.} jointly optimized hybrid beamforming and RIS phase-shift matrices to improve system sensing performance. 
In \cite{10143420}, Hua {\it et al.} addressed physical layer security issues by collaboratively designing active and passive beamforming to effectively suppress information leakage while ensuring communication quality. 
In \cite{10901330}, Liu {\it et al.} discussed device-based sensing scenarios and investigated the uplink signal transmission mechanism of active targets aided by RIS.

However, traditional RIS employs the diagonal phase-shift matrix structure, where reflecting elements are independent of each other, which limits the beamforming gain in massive multiple input multiple output (MIMO) systems~\cite{10316535}. 
To address this limitation, the beyond diagonal reconfigurable intelligent surface (BD-RIS) architecture has been proposed. Unlike conventional architectures, BD-RIS allows reflecting elements to be interconnected through internal impedance networks, thereby giving rise to two novel architectures: fully-connected BD-RIS (FBD-RIS) and group-connected BD-RIS (GBD-RIS)~\cite{10453384,10197228}.
In the fully-connected architecture, all reflecting elements are mutually interconnected, with their reflection characteristics characterized by a full-rank matrix. The group-connected architecture partitions the reflecting elements into multiple independent subgroups, where elements within each subgroup remain fully connected, thus forming a block-diagonal reflection matrix. 
The FBD-RIS can achieve higher gains, but correspondingly entails greater circuit implementation complexity. In contrast, the GBD-RIS offers performance that lies between conventional diagonal RIS (D-RIS) and FBD-RIS, while maintaining relatively moderate complexity~\cite{10766364}.
Currently, research has begun to explore BD-RIS-aided ISAC systems. In \cite{10493847}, Guang {\it et al.} studied the minimization of base station (BS) transmit power under the conditions of meeting communication and sensing quality constraints. In \cite{10495009}, Liu {\it et al.} maximized system throughput under the premise of ensuring sensing performance.

Different from existing work, this paper proposes a novel multiuser interference management and beam gain approximation method for BD-RIS-aided ISAC systems, aiming to achieve the flexible design of trade-offs between communication and sensing performance. Specifically, the proposed interference management method can improve multiuser communication quality, while the beam gain approximation method can enhance sensing beam gain. 
The advantage of this method lies in transforming the originally highly non-convex objective function into a convex function form, avoiding non-convex quartic sensing constraints or communication signal-to-interference-plus-noise ratio (SINR) constraints, thereby significantly reducing the computational complexity of the optimization algorithm. 
For the non-convex weighted optimization problem formed thereby, we employ an alternating optimization (AO) algorithm to decompose it into two sub-problems of precoding vector and phase-shift matrix optimization for alternating solution. 
These subproblems are solved iteratively, with each iteration admitting closed-form solution, which further enhances computational efficiency.
Simulation results show that compared with traditional D-RIS-aided schemes, BD-RIS-aided ISAC systems exhibit obvious advantages in terms of the trade-offs between communication and sensing performance.

{\it{Notation:}} 
Scalars are denoted by italic letters, while vectors and matrices are represented by bold-face lowercase and uppercase letters, respectively. $ \|\cdot\| $ represents the Euclidean norm of a vector, and $ |\cdot| $ denotes the absolute value of a scalar. $ \bm{x} \sim \mathcal{CN}(\bm{0}, \mathbf{X}) $ indicates that $ \bm{x} $ follows a circularly symmetric complex Gaussian (CSCG) distribution with mean vector $ \bm{0} $ and covariance matrix $ \mathbf{X} $. 
$ (\cdot)^{\rm H} $, $ (\cdot)^{\rm T} $, $ (\cdot)^{*} $, and $ (\cdot)^{-1} $ refer to the conjugate transpose, transpose, and conjugate operation, matrix inverse operation, respectively.
$\otimes$ and $\odot$ denote the Kronecker product and the Hadamard product, respectively. $\mathrm{Re}\{\cdot\}$ and $\mathrm{Im}\{\cdot\}$ denote the real and imaginary parts, respectively.
$\mathbf{I}_{N \times N}$ denotes the $N\times N$ dimensional matrix, and $ \mathbb{C}^{M \times N} $ denotes the $ M \times N $ dimensional space of complex matrices. 
$ \text{Re}(\cdot) $ represents the real part of a complex number, and $ \mathbb{E}[x] $ denotes the statistical expectation (or mean) of the random variable $ x $.
$ \mathrm{diag}(\bm{x}) $ refers to the diagonal matrix formed from the components of the vector $ \bm{x} $, and $ x_n $ denotes the $ n $-th element of the vector $ \bm{x} $. $ j $ is the imaginary unit.

\section{System Model}
\label{s2}
We consider a BD-RIS-aided ISAC system, as illustrated in Fig.~\ref{model}. In the above system, an $M$-antenna BS simultaneously serves $K$ single-antenna users indexed by $\mathcal{K} =
\{1,2,...,K\}$ while sensing a point-like target by BD-RIS equipped with $N=N_1 \times N_2$ reflecting elements.

\begin{figure}[!ht]
	\centering
	\includegraphics[width=0.45\textwidth]{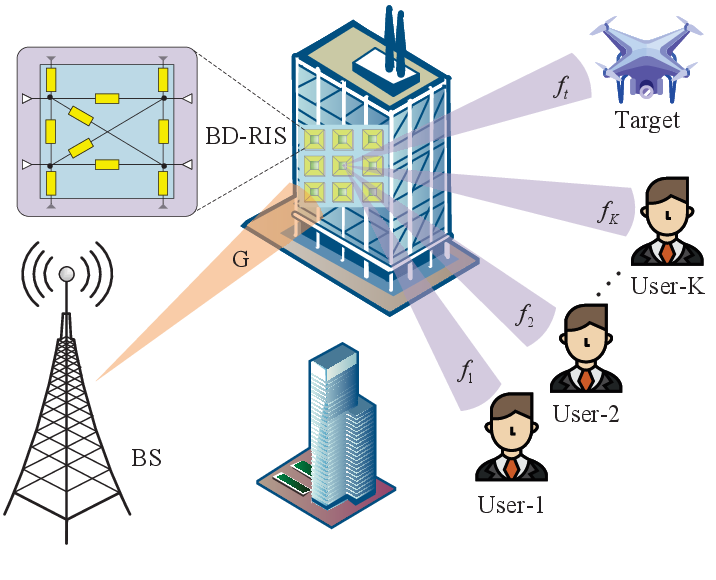}
	\caption{The model of a BD-RIS-aided ISAC system.}
	\label{model}
\end{figure}

\subsection{Channel Model}
In the above ISAC system, we denote $\mathbf{G}\in\mathbb{C}^{N\times M}$, $\bm f_K\in\mathbb{C}^{N\times K}$ and $\bm f_t\in\mathbb{C}^{N\times 1}$ as the channels from BS to BD-RIS, from the RIS to user-$k$ and from the RIS to target, respectively. We assume that the direct links from the BS to users and target are blocked due to severe path loss or physical obstructions. 
The channel $\bm f_k$ is modeled as the Rayleigh channel model, expressed as
\begin{equation} 
	\begin{split}
		\bm{f}_k = \sqrt{\beta_k} \mathbf{R}^{1/2}_{\text{RU}} \bm{z}_k,
	\end{split}
\end{equation}

\noindent where $\beta_k$ denotes the path loss, $\mathbf{R}_{\text{RU}} \in \mathbb{C}^{N \times N}$ represents the spatial correlation matrix, and $\bm{z}_k \sim \mathcal{CN}(\mathbf{0}, \mathbf{I}_N)$ describes the corresponding fast-fading vectors,
\begin{equation} 
	\begin{split}
		[\mathbf{R}_{\text{RU}}]_{i,j} = \text{sinc}\left( \frac{2\|\bm{u}_i - \bm{u}_j\|_2}{\lambda} \right),
	\end{split}
\end{equation}
\noindent where $\bm{u}_i = [0, \text{mod}(i-1, N_1)\lambda/2, \lfloor(i-1)/N_1\rfloor\lambda/2]$ denotes the position of the $i$-th RIS element, $\lambda$ is the wavelength, and $i,j \in \mathcal{N} \triangleq \left\{1,2,\ldots, N\right\}$~\cite{9300189,10606173}.

For the sensing link, we model the channel $\bm f_t$ as the line-of-sight (LOS) steering vector, which is given by
\begin{equation} 
	\begin{split}
		\bm{f}_t = &\left[1, e^{-j\pi\sin(\theta^t)\sin(\varphi^t)}, \ldots, e^{-j\pi(N_1-1)\sin(\theta^t)\sin(\varphi^t)}\right]^{\mathrm{T}} \\ 
		&\otimes \left[1, e^{-j\pi\cos(\theta^t)}, \ldots, e^{-j\pi(N_2-1)\cos(\theta^t)}\right]^{\mathrm{T}} / \sqrt{N},
	\end{split}
\end{equation}
\noindent where $\theta^t$ and $\varphi^t$ denote the elevation and azimuth angles of the target, respectively.

Furthermore, due to its fully-connected impedance networking architecture, the BD-RIS phase-shift matrix $\mathbf{\Theta} \in \mathbb{C}^{N\times N}$ should satisfy 
$\mathbf{\Theta}^\mathrm{H} \mathbf{\Theta} = \mathbf{I}_N$ and $\mathbf{\Theta}^\mathrm{T} = \mathbf{\Theta}$~\cite{9913356}.

\subsection{Communication and Sensing Model}
In the considered ISAC system, the BS transmits a dual-functional signal that simultaneously can be used for communication with users and sensing of the target. 
Accordingly, the transmitted signal from the BS can be formulated as 
\begin{equation}
	\bm x = \mathbf{W} \bm s = \sum_{k=1}^K{\bm w_k s_k},
\end{equation}

\noindent where $\mathbf{W} = [\bm{w}_1, \bm{w}_2, \ldots, \bm{w}_K] \in \mathbb{C}^{M \times K}$ represents the beamforming matrix, with $\bm{w}_k \in \mathbb{C}^{M \times 1}$ denoting the beamforming vector for user $k$.
$\bm s=[s_1,\dots,s_K]^\mathrm{T}$ contains the data symbols intended for all $K$ users, where the information symbols satisfy the normalization constraint $\mathbb{E} \left[\bm s \bm s^\mathrm{H}\right] = \mathbf{I}$.

Therefore, the received signal at user-$k$ can be expressed as
\begin{equation}
	\begin{split}\label{y_k}
		y_k = \underbrace{\bm{f}_k^\mathrm{H} \mathbf{\Theta} \mathbf{G} \bm{w}_k s_k}_{\text{desired signal}} + \underbrace{\sum_{i=1, i \neq k}^{K} \bm{f}_k^\mathrm{H} \mathbf{\Theta} \mathbf{G} \bm{w}_i s_i}_{\text{multiuser interference}} + \underbrace{n_k}_{\text{noise}},
	\end{split}
\end{equation}

\noindent where $ n_k \sim \mathcal{C} \mathcal{N}\left(0,\sigma_k^2\right)$ denotes the additive white Gaussian noise (AWGN).

From \eqref{y_k}, it can be observed that the successful decoding of user-$k$ depends on two key factors: enhancing the desired signal component and suppressing the interference from other users. 
To facilitate the joint optimization of beamforming and RIS phase shifts for enhanced communication performance, we introduce the multiuser beam gain matrix defined as
\begin{equation} 
	\mathbf{F} = (\mathbf{HW}) \odot (\mathbf{HW})^*, 
\end{equation}
\noindent where 
$\mathbf{H} = \begin{bmatrix}
	\bm{f}_1^\mathrm{H} \mathbf{\Theta} \mathbf{G}, \bm{f}_2^\mathrm{H} \mathbf{\Theta} \mathbf{G}, \ldots, \bm{f}_K^\mathrm{H} \mathbf{\Theta} \mathbf{G}
\end{bmatrix}^\mathrm{T} \in \mathbb{C}^{K \times M}$. 
Ideally, to suppress inter-user interference and enhance the SINR, the beam gain matrix $\mathbf{F}$ should approximate diagonal structure. This observation motivates the design criterion of minimizing the difference between $\mathbf{F}$ and the diagonal matrix $\mathbf{P}$, i.e., $\min_{\mathbf{W}, \mathbf{\Theta}} \|\mathbf{F} - \mathbf{P}\|_F^2$, which effectively enhance the communication performance.

Furthermore, to effectively sense the point-like target in the interested direction, sufficient beam gain must be achieved toward the sensing direction~\cite{8386661}. 
The sensing beam gain in the target direction can be characterized as
\begin{equation}
	F_s = \left|\bm{f}_t^\mathrm{H}  \mathbf{\Theta} \mathbf{G} \mathbf{W}\right|^2 = \left|\bm{f}_t^\mathrm{H}  \mathbf{\Theta} \mathbf{G} \sum_{k=1}^{K}{\bm{w}_k}\right|^2.
\end{equation}
Therefore, the objective can be accomplished by minimizing the difference between the sensing beam gain and the desired beam pattern, i.e., 
$\min_{\mathbf{W}, \mathbf{\Theta}} |F_s - P_s|^2,$
where $P_s \!\geq\! 0$ is a given desired gain.

\section{Beamforming Optimization}

In this section, we consider the design of trade-offs between communication and sensing. Therefore, we formulate the following weighted optimization problem
\begin{subequations}\label{weighted_ISAC}
	\begin{align}
		\min_{\mathbf{W}, \mathbf{\Theta}} \quad & f(\mathbf{W}, \mathbf{\Theta}) = \eta \|\mathbf{F} - \mathbf{P}\|_F^2 + (1-\eta) |{F}_s - {P}_s|^2	\label{eq:objective_main}\\
		\text{s.t.} \quad & \mathcal{C}_{\text{BS}}: \sum_{k=1}^K \|\bm{w}_k\|^2_2 \leq P_{\max} \label{eq:constraint_power},\\
		& \mathcal{C}_{\mathbf{\Theta}}: 
		\mathbf{\Theta}^\mathrm{H} \mathbf{\Theta} = \mathbf{I}_N,
		\mathbf{\Theta}^\mathrm{T} = \mathbf{\Theta},	\label{theta_constraint}		
	\end{align}
\end{subequations}
\noindent where $\eta \in [0,1]$ is a given weighted factor. $P_{\max}$ represents the maximum transmit power at the BS.

Problem \eqref{weighted_ISAC} is challenging to solve directly due to the non-convexity of the objective function and constraint \eqref{theta_constraint}. In particular, the variables $\mathbf{W}$ and $\mathbf{\Theta}$ appear in the objective as fourth-order terms $\|\mathbf{W}\|_F^4$ and $\|\mathbf{\Theta}\|_F^4$, which further complicates the optimization.

Therefore, we introduce an equivalent transformation of problem \eqref{weighted_ISAC} as follows
\begin{subequations}\label{weighted_ISAC_2}
	\begin{align}
		\min_{\bm{w}, \mathbf{\Theta}, \bm \theta, \bm \phi} \quad & f(\bm{w}, \mathbf{\Theta},\bm \theta, \bm \phi) = \eta\sum_{k=1}^K \sum_{i=1}^K \left| \bm{f}_i^\mathrm{H} \mathbf{\Theta} \mathbf{G} \bm{w}_k - P_{i,k}e^{j\theta_{i,k}} \right|^2 \notag\\
		& \quad + (1-\eta)\sum_{k=1}^K \left| \bm{f}_t^\mathrm{H} \mathbf{\Theta} \mathbf{G} \bm{w}_k - P_t e^{j\phi_k} \right|^2\\
		\text{s.t.} \quad & \mathcal{C}_{\text{BS}}: \sum_{k=1}^K \|\bm{w}_k\|^2_2 \leq P_{\max} ,\\
		& \mathcal{C}_{\mathbf{\Theta}}: 
		\mathbf{\Theta}^\mathrm{H} \mathbf{\Theta} = \mathbf{I}_N,
		\mathbf{\Theta}^\mathrm{T} = \mathbf{\Theta},			
	\end{align}
\end{subequations}

\noindent where $\bm \theta \!=\! [\theta_{1,1},\dots,\theta_{i,k},\dots,\theta_{K,K}]^\mathrm{T}, \bm \phi \!=\! [\phi_{1},\phi_{2},\dots,\phi_{K}]^\mathrm{T}$, $P_{i,k} = C$ for $i = k$ and $P_{i,k} = 0$ for $i \neq k$.

Problem \eqref{weighted_ISAC_2} is very tricky due to the coupled variables and its non-convex constraint.  To tackle these difficulties, we develop an efficient alternating optimization (AO) algorithm.

\subsection{Beamforming Optimization}
When $\mathbf{\Theta}, \bm \theta$ and $\bm \phi$ are fixed, the optimization problem about $\bm w$ can be reduced to the following problem
\begin{equation}
	\min_{\{\bm{w}_k\}_{k=1}^K} \ f(\bm w) \quad \text{s.t.} \ \mathcal{C}_{\text{BS}}.
	\label{subproblem1}
\end{equation}

Then, problem \eqref{subproblem1} can be equivalently reformulated as
\begin{subequations}
	\begin{align}\label{subproblem1_2}
	\min_{\{\bm{w}_k\}_{k=1}^K} \quad &f(\bm w) = \sum_{k=1}^K \left( \bm{w}_k^\mathrm{H} \mathbf{A} \bm{w}_k - 2\text{Re}\{\bm{b}_k^\mathrm{H} \bm{w}_k\}\right)\\
	\text{s.t.} \quad & \mathcal{C}_{\text{BS}}: \sum_{k=1}^K \|\bm{w}_k\|^2_2 \leq P_{\max},
	\end{align}
\end{subequations}

\noindent where 
\begin{equation*}
	\mathbf{A} = \eta\sum_{i=1}^K \mathbf{G}^\mathrm{H} \mathbf{\Theta}^\mathrm{H} \bm{f}_i \bm{f}_i^\mathrm{H} \mathbf{\Theta} \mathbf{G} + (1-\eta) \mathbf{G}^\mathrm{H} \mathbf{\Theta}^\mathrm{H} \bm{f}_t \bm{f}_t^\mathrm{H} \mathbf{\Theta} \mathbf{G},
\end{equation*}
\begin{equation*}
	\bm{b}_k = \eta\sum_{i=1}^K P_{i,k}e^{j\theta_{i,k}} \mathbf{G}^\mathrm{H} \mathbf{\Theta}^\mathrm{H} \bm{f}_i + (1-\eta) P_t e^{j\phi_k} \mathbf{G}^\mathrm{H} \mathbf{\Theta}^\mathrm{H} \bm{f}_t.
\end{equation*}

In order to solve the above problem, we derive the semi-closed form solution of the problem \eqref{subproblem1_2} by the Lagrangian method as
\begin{equation}
	\bm{w}_k^\star = (\mathbf{A} + \lambda_k \mathbf{I})^{-1} \bm{b}_k, \quad k \in \mathcal{K},
\end{equation}
where the Lagrange multiplier $\lambda_k \geq 0$ must satisfy the complementary slackness condition $\lambda_k(\|\bm{w}_k\|^2_2 - P_{\max}) = 0$, which can be efficiently obtained via bisection search.

\subsection{Phase-Shift Matrix Optimization}
In this subsection, we fix $\bm{w}$, $\bm{\theta}$ and $\bm{\phi}$, and optimize the RIS phase-shift matrix $\mathbf{\Theta}$. 
Due to the complex matrix structure of $\mathbf{\Theta}$, directly solving this problem is intractable. 
Therefore, we introduce an auxiliary variable $\bm{\psi} = \text{vec}(\mathbf{\Theta})$ to decouple $\mathbf{\Theta}$ from the complicated constraints. 
Then, the optimization problem can be reformulated as scaled-form augmented Lagrangian
\begin{equation}\label{subproblem2}
	\min_{\mathbf{\Theta}} \ f(\mathbf{\Theta}) +  \mu\left\|\bm{\psi} - \text{vec}(\mathbf{\Theta}) + \bm{\nu}\right\|^2_2 \quad \text{s.t.} \  \mathcal{C}_{\mathbf{\Theta}},
\end{equation}

\noindent where $\mu > 0$ and $\bm{\nu} \in \mathbb{C}^{N^2 \times 1}$ denote the given parameter and the dual variable, respectively. To solve the problem \eqref{subproblem2}, we employ the following update rules.
\begin{subequations}
	\begin{align}
		\bm{\psi}^{l+1} =& \ \arg\min_{\bm{\psi}} \ f(\bm{\psi}) + \mu\left\|\bm{\psi} - \text{vec}(\mathbf{\Theta}^{l}) + \bm{\nu}^{l}\right\|^2_2,\label{subproblem_1}\\
		\mathbf{\Theta}^{l+1} =& \ \arg\min_{\mathbf{\Theta}} \left\|\text{vec}(\mathbf{\Theta}) - \bm{\psi}^{l+1} - \bm{\nu}^{l}\right\|^2_2, \notag\\
		& \text{s.t. } \mathbf{\Theta}^\mathrm{H} \mathbf{\Theta} = \mathbf{I}_N, \mathbf{\Theta}^\mathrm{T} = \mathbf{\Theta},\label{subproblem_2}\\
		\bm{\nu}^{l+1} =& \ \bm{\nu}^{l} + \bm{\psi}^{l+1} - \text{vec}(\mathbf{\Theta}^{l+1}).
	\end{align}
\end{subequations}

For subproblem \eqref{subproblem_1}, it can be equivalently formulated as
\begin{equation}
	\min_{\bm{\psi}} \ \bm{\psi}^\mathrm{H} (\mathbf{P}+\mathbf{Q}) \bm{\psi} - 2\text{Re}\{(\bm{p}+\bm{q})^\mathrm{H} \bm{\psi}\} + \mu\left\|\bm{\psi} - \text{vec}(\mathbf{\Theta}) + \bm{\nu}\right\|^2_2,\label{subproblem_1_1}
\end{equation}

\noindent where
\begin{subequations}
	\begin{align*}
	\mathbf{P} &= \eta\sum_{i=1}^K \sum_{k=1}^K (\bm{w}_k^* \mathbf{G}^* \otimes \bm{f}_i)(\bm{w}_k^\mathrm{T} \mathbf{G}^\mathrm{T} \otimes \bm{f}_i^\mathrm{H}) ,\\
	\mathbf{Q} &= (1-\eta)\sum_{k=1}^K (\bm{w}_k^* \mathbf{G}^* \otimes \bm{f}_t)(\bm{w}_k^\mathrm{T} \mathbf{G}^\mathrm{T} \otimes \bm{f}_t^\mathrm{H}),\\
	\bm{p} &= \eta\sum_{i=1}^K \sum_{k=1}^K P_{i,k}e^{j\theta_{i,k}} (\bm{w}_k^* \mathbf{G}^* \otimes \bm{f}_i) ,\\
	\bm{q} &= (1-\eta)\sum_{k=1}^K P_t e^{j\phi_k} (\bm{w}_k^* \mathbf{G}^* \otimes \bm{f}_t).
	\end{align*}
\end{subequations}

Then, the closed-form solution of the problem \eqref{subproblem_1_1} can be derived as
\begin{equation}
	\bm{\psi}^\star = (\mathbf{P} + \mathbf{Q} + \mu \mathbf{I})^{-1} \left(\bm{p} + \bm{q} + \mu (\text{vec}(\mathbf{\Theta}) - \bm{\nu})\right).
\end{equation}

Furthermore, subproblem \eqref{subproblem_2} can be equivalently reformulated as the following projection problem
\begin{equation}
	\min_{\boldsymbol{\Theta}} \|\boldsymbol{\Theta} - \hat{\boldsymbol{\Theta}}\|_2^2 
	\quad \text{s.t.} \ \boldsymbol{\Theta}^\mathrm{H}\boldsymbol{\Theta} = \mathbf{I},  \boldsymbol{\Theta}^\mathrm{T} = \boldsymbol{\Theta},
	\label{12}
\end{equation}

\noindent where $\hat{\boldsymbol{\Theta}} = \text{mat}(\boldsymbol{\theta}^{m+1} + \boldsymbol{\nu}^m)$ denotes the matricization of vector $\boldsymbol{\theta}^{m+1} + \boldsymbol{\nu}^m$ into an $N \times N$ matrix.

For the optimization problem \eqref{12}, a closed-form solution can be derived through symmetric unitary projection.
\begin{equation}
	\boldsymbol{\Theta}^\star = \text{symuni}(\hat{\boldsymbol{\Theta}}) = \tilde{\mathbf{U}}\mathbf{V}^\mathrm{H},
\end{equation}
where $\tilde{\mathbf{U}} = [\mathbf{U}_{:,1:R}, \mathbf{V}^*_{:,R+1:N}]$, $R = \text{rank}(\text{sym}(\hat{\boldsymbol{\Theta}}))$.
$\mathbf{U}$ and $\mathbf{V}$ are the unitary matrices obtained from the singular value decomposition (SVD) of $\text{sym}(\hat{\boldsymbol{\Theta}}) = (\hat{\boldsymbol{\Theta}} + \hat{\boldsymbol{\Theta}}^\mathrm{T})/2$~\cite{10319662,10742100}.

\subsection{Desired Beam Phase Optimization}
When $\bm w$ and $\mathbf{\Theta}$ are fixed, the optimization problem about $\bm \theta$ and $\bm \phi$ can be reduced to the following problems
\begin{equation}
	\min_{\bm \theta, \bm \phi} \ f(\bm \theta, \bm \phi) \quad \text{s.t.} \ \theta_{i,k}\in [0, 2\pi), \phi_{k}\in [0, 2\pi).
\end{equation}

The variables $\bm\theta$ and $\bm \phi$ are determined by solving the following subproblems in parallel
\begin{subequations}\label{theta}
	\begin{align}
	&\min_{\theta_{i,k}} \ \left| \bm{f}_i^\mathrm{H} \mathbf{\Theta} \mathbf{G} \bm{w}_k - P_{i,k}e^{j\theta_{i,k}} \right|^2 \quad \text{s.t.} \ \theta_{i,k} \in [0, 2\pi),\\
	&\min_{\phi_{k}} \ \left| \bm{f}_t^\mathrm{H} \mathbf{\Theta} \mathbf{G} \bm{w}_k - P_{t}e^{j\phi_{k}} \right|^2 \quad \text{s.t.} \ \phi_{k} \in [0, 2\pi).
	\end{align}
\end{subequations}

To tackle the subproblems \eqref{theta}, we formulate the following proposition along with its corresponding solution approach.

\begin{proposition}\label{prop:phase_opt}
	Consider the following optimization problem
	\begin{equation}\label{eq:phase_opt_problem}
		\min_{\xi \in [0, 2\pi)} \left|e^{j\xi} - \zeta\right|^2,
	\end{equation}
	where $\zeta$ is a given value. The globally optimal solution of problem \eqref{eq:phase_opt_problem} can be derived as
	\begin{equation}\label{eq:phase_opt_solution}
		\xi^\star = \text{mod}\left(\angle(\zeta), 2\pi\right).
	\end{equation}
\end{proposition}

\begin{proof}[\bfseries Proof]
	The problem \eqref{eq:phase_opt_problem} can be written as
	\begin{equation}
		\min_{\xi \in [0, 2\pi)} \left(\sin(\xi) - \text{Im}\{\zeta\}\right)^2 + \left(\cos(\xi) - \text{Re}\{\zeta\}\right)^2.
	\end{equation}
		
	Furthermore, the problem \eqref{eq:phase_opt_problem} can be equivalently reformulated as
	\begin{equation}	
		\max_{\xi \in [0, 2\pi)} \text{Re}\{\zeta\}\cos(\xi) + \text{Im}\{\zeta\}\sin(\xi) .
	\end{equation}

	Due to $\text{Re}\{\zeta\} = |\zeta|\cos(\angle(\zeta))$, $\text{Im}\{\zeta\} = |\zeta|\sin(\angle(\zeta))$, and $\cos(\angle(\zeta))\cos(\xi) + \sin(\angle(\zeta))\sin(\xi) = \cos(\angle(\zeta) - \xi)$, we can obtain 
	\begin{equation}
		\max_{\xi \in [0, 2\pi)} \cos(\angle(\zeta) - \xi).
	\end{equation}
	
	Therefore, the globally optimal solution of the problem \eqref{eq:phase_opt_problem} can be derived as $\xi^\star = \bmod(\angle(\zeta), 2\pi)$.

	The proof is completed.
\end{proof}

According to \textit{\textbf{Proposition 1}}, the optimal solution of problem \eqref{theta} is given by
\begin{equation}
	\theta_{i,k}^\star = \angle\left(\bm{f}_i^\mathrm{H} \mathbf{\Theta} \mathbf{G} \bm{w}_k\right), \ \phi_k^\star = \angle\left(\bm{f}_t^\mathrm{H} \mathbf{\Theta} \mathbf{G} \bm{w}_k\right).
	\label{eq:optimal_solution}
\end{equation}

It is worth noting that the proposed algorithm exhibits excellent extensibility. By simply applying the projection operation in \eqref{12} to each diagonal block of the phase-shift matrices, i.e., $\text{diag}(\text{symuni}(\boldsymbol{\Theta}_1), \text{symuni}(\boldsymbol{\Theta}_2), \ldots, \text{symuni}(\boldsymbol{\Theta}_G))$, it can be directly applied to the proposed the design of trade-offs in the GBD-RIS-aided ISAC system.

\section{simulation results}
\label{s4}
In this section, simulations results are conducted to validate the performance of the proposed algorithm. We consider a typical ISAC scenario where the BS and the BD-RIS are located at $(-20\text{m}, 0\text{m}, 25\text{m})$ and $(0\text{m}, 0\text{m}, 0\text{m})$, respectively. The key system parameters are configured as $K = 5$, $M = 8$, $N_1 = 8$, $N_2 = 4$, $\lambda = 0.03\text{m}$,   $P_{\max} = 30\text{dBm}$,  $\theta^t = 90^{\circ}$, $\phi^t = 45^{\circ}$, and $\sigma_k^2 = -100~\text{dBm}$, respectively. The path loss coefficient is modeled as $\beta_k = 10^{-3}d_k^{-2}$, where $d_k$ denotes the distance from the RIS to the $k$-th user. 
The small-scale fading is characterized by Rician channel model, while the large-scale 
fading for the BS-RIS channel is described by $\mathrm{PL} = 37.3+22.0\log\left(d\right)$, where $d$ denotes the distance between BS and RIS~\cite{10742100}. 




\begin{figure*}[!htbp]
	\centering
	\subfloat[$\eta = 0$]{
		\includegraphics[width=0.3\linewidth]{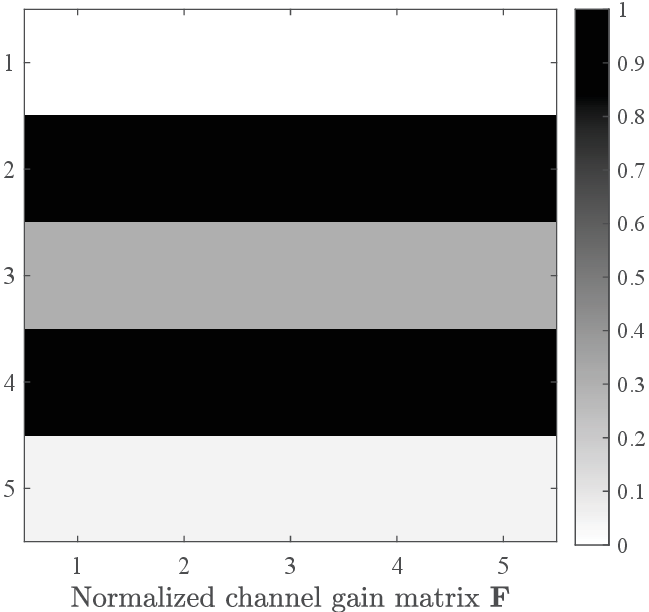}
		\label{HWI}}
	\hspace{0cm}
	\subfloat[$\eta = 0.6$]{
		\includegraphics[width=0.3\linewidth]{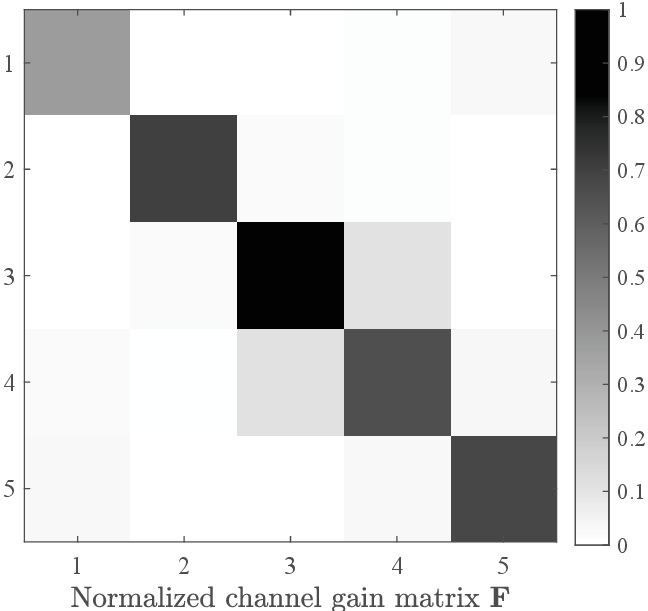}
		\label{SIC}}%
	\hspace{0cm}
	\subfloat[$\eta = 1$]{
		\includegraphics[width=0.3\linewidth]{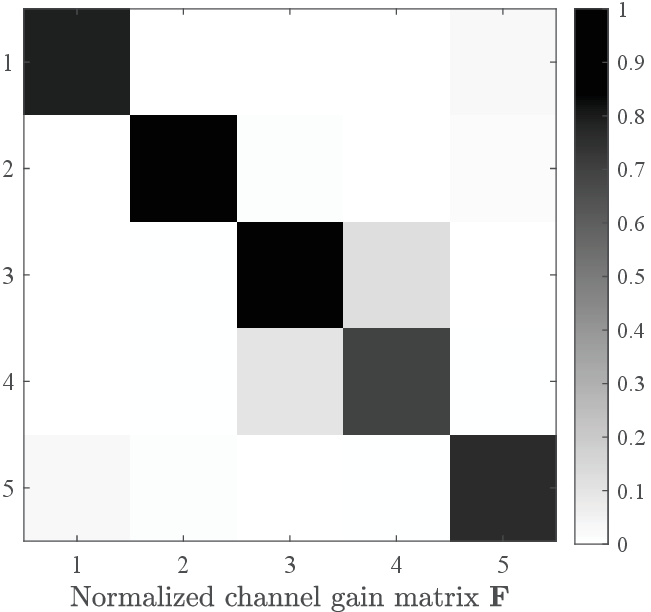}
		\label{RIS}}
	\hspace{0cm}
	\caption{User channel gain matrix $\mathbf{F}$ generated for different weighted values $\eta$,}
	\label{communication}
\end{figure*}

\begin{figure*}[!htbp]
	\centering
	\subfloat[$\eta = 0$]{
		\includegraphics[width=0.3\linewidth]{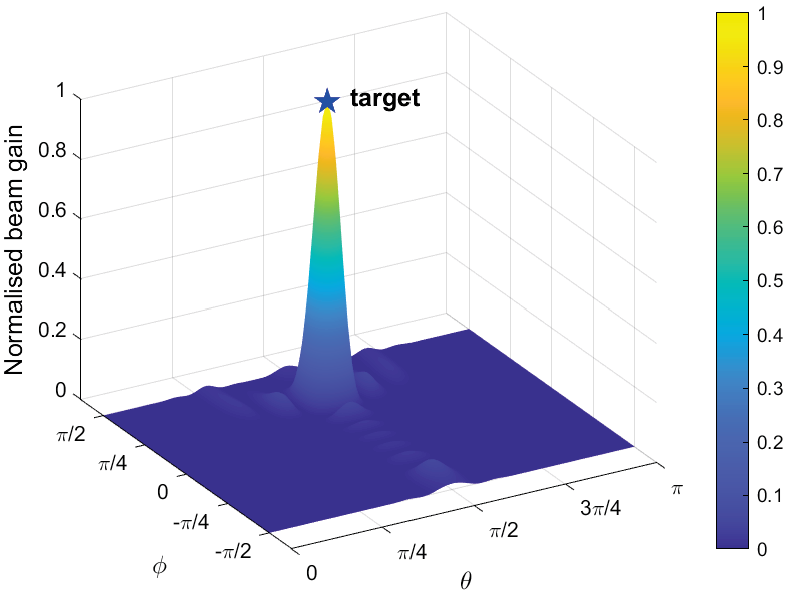}
		\label{HWI}}
	\hspace{0cm}
	\subfloat[$\eta = 0.6$]{
		\includegraphics[width=0.3\linewidth]{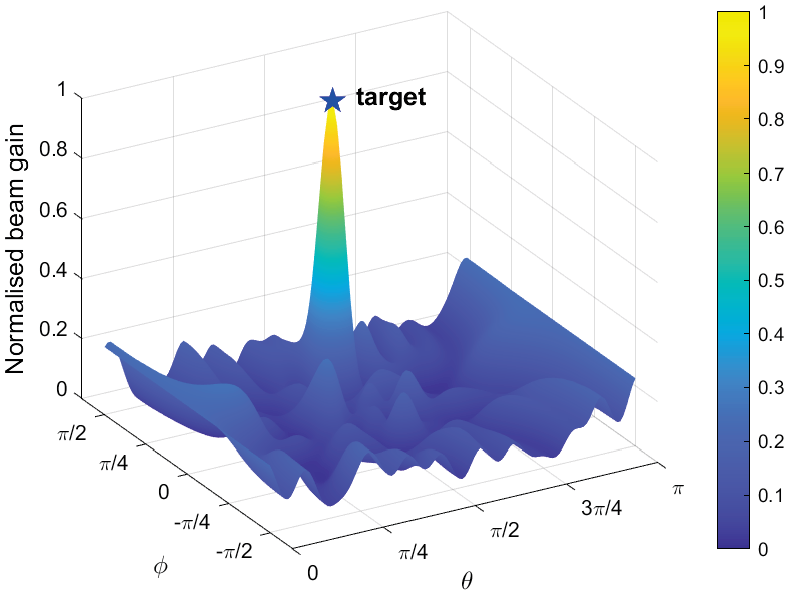}
		\label{SIC}}%
	\hspace{0cm}
	\subfloat[$\eta = 1$]{
		\includegraphics[width=0.3\linewidth]{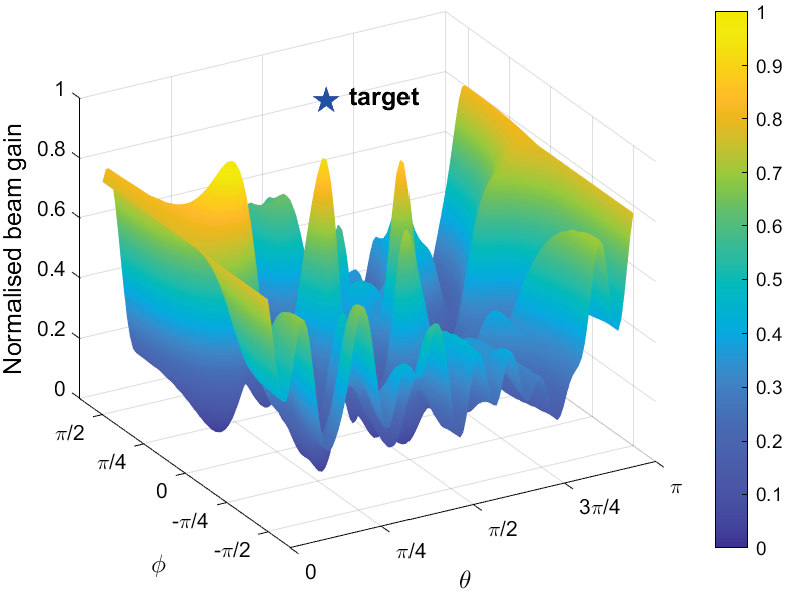}
		\label{RIS}}
	\hspace{0cm}
	\caption{Beam pattern generated for different weighted values $\eta$.}
	\label{sensing}
\end{figure*}

Fig.~\ref{communication} and Fig.~\ref{sensing} illustrate the channel gain matrix $\mathbf{F}$ and the corresponding beam patterns for different weights, respectively. As observed from Fig.~\ref{communication}, as the $\eta$ increases from 0 to 1, the system design gradually transitions from sensing-centered to communication-centered, and the channel gain matrix $\mathbf{F}$ exhibits increasingly pronounced diagonalization characteristics, indicating that BD-RIS can enhance the orthogonality among user channels through phase-shift optimization, thereby effectively suppressing multiuser interference and improving communication performance. Meanwhile, as shown in Fig.~\ref{sensing}, as $\eta$ decreases, the system design becomes more sensing-centered, and the beam gain is more concentrated toward the direction of interest, thereby achieving superior sensing performance. Therefore, in practical applications, BD-RIS-aided ISAC systems can flexibly adjust $\eta$ to adapt to different sensing and communication performance requirements under various scenarios, achieving higher-performance ISAC systems.

\begin{figure}[!t]
	\centering
	\includegraphics[width=0.5\textwidth]{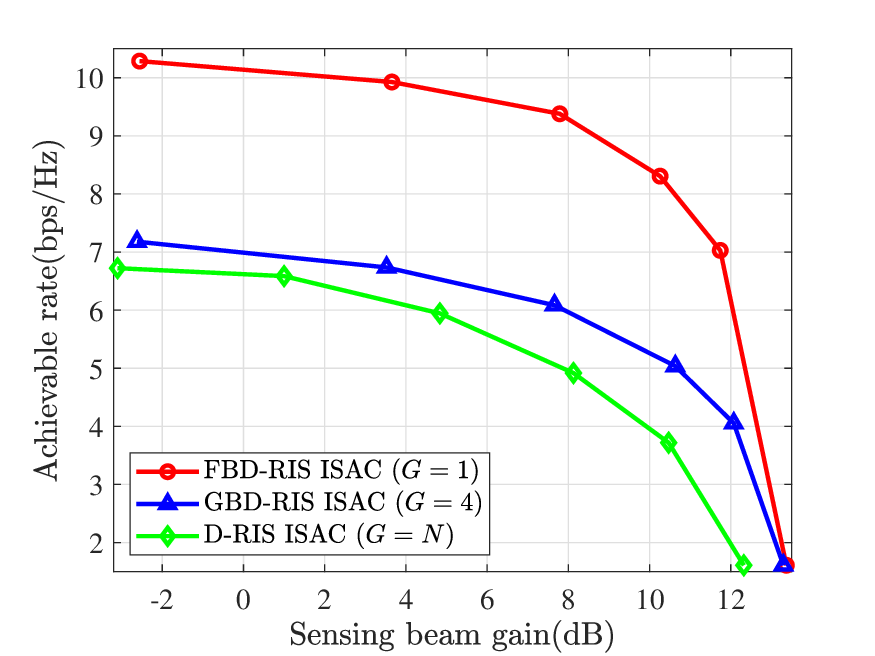}
	\caption{The trade-offs between the achievable rate and sensing beam gain.}
	\label{Trade_offs}
\end{figure}

Fig.~\ref{Trade_offs} illustrates the design of trade-offs between the achievable rate and sensing beam gain in ISAC systems aided by different RIS architectures. 
The simulation results demonstrate that all RIS architectures exhibit a monotonically decreasing trend in achievable rate as the sensing beam gain increases, which reveals the inherent trade-offs between communication and sensing performance in ISAC systems.
Specifically, when more power resources are allocated to enhance sensing performance, the power budget available for communication is correspondingly reduced, leading to a degradation in multiuser sum rate.
Moreover, the FBD-RIS-aided ISAC system consistently achieves superior system performance, significantly outperforming the GBD-RIS and conventional D-RIS-aided ISAC systems, and this performance advantage primarily stems from the fact that FBD-RIS can provide more spatial degrees of freedom (DoF). The simulation results substantiate that employing the BD-RIS can effectively enhance the overall performance of ISAC systems under the design of trade-offs between communication and sensing scenarios.

\section{conclusion}
\label{s5}
In this paper, we have investigated the optimization problem of BD-RIS-aided multiuser ISAC systems and have proposed an optimization framework jointly combining multi-user interference management method and sensing beam gain approximation. For the non-convex weighted optimization problem formulated in this paper, an alternating optimization algorithm has been employed to solve the base station precoding and BD-RIS phase-shift matrix, combined with the alternating direction method of multipliers to achieve efficient solution. Simulation results have demonstrated that compared with traditional diagonal RIS schemes, the BD-RIS-aided ISAC system achieves significant advantages in trade-offs between communication and sensing. Future work will explore robust design under imperfect channel information and extended applications in multi-target sensing scenarios.

\bibliographystyle{IEEEtran}
\bibliography{myref}

\end{document}